\documentclass[journal]{IEEEtran}
\makeatletter
\def\ps@headings{%
\def\@oddhead{\mbox{}\scriptsize\rightmark \hfil \thepage}%
\def\@evenhead{\scriptsize\thepage \hfil \leftmark\mbox{}}%
\def\@oddfoot{}%
\def\@evenfoot{}}
\makeatother
\usepackage{nonfloat}
\usepackage{subfig}
\usepackage{multirow}
\usepackage{graphicx}
\usepackage{amssymb}
\usepackage{epstopdf}
\usepackage{amsmath}
\usepackage{amsthm}
\usepackage{cite}
\newtheorem{mydef}{Definition}
\newtheorem{lem}{Lemma}
\newtheorem{thm}{Theorem}
\newtheorem{cor}{Corollary}

\begin{document}
\title{An Upper Bound on the Convergence Time for Quantized Consensus of Arbitrary Static Graphs}

%
       
       \author{
\IEEEauthorblockN{Shang Shang\IEEEauthorrefmark{1},  Paul Cuff\IEEEauthorrefmark{1}, Pan
Hui\IEEEauthorrefmark{2}
and Sanjeev Kulkarni\IEEEauthorrefmark{1} }\\ \IEEEauthorblockA{\IEEEauthorrefmark{1}Department of Electrical Engineering, 
Princeton University \\ Princeton NJ, 08540, U.S.A.
} \\ \IEEEauthorblockA{\IEEEauthorrefmark{2}Department of Computer Science and Engineering, The Hong Kong University of Science and Technology, \\
Hong Kong\\
\IEEEauthorrefmark{1}\{sshang, cuff, kulkarni\}@princeton.edu, \IEEEauthorrefmark{2}panhui@cse.ust.hk}
}

\maketitle
\begin{abstract}
We analyze a class of distributed quantized consensus algorithms for arbitrary static networks. In the initial setting, each node in the network has an integer value. Nodes exchange their current estimate of the mean value in the network, and then update their estimation by communicating with their neighbors in a limited capacity channel in an asynchronous clock setting. Eventually, all nodes reach consensus with quantized precision. We analyze the expected convergence time for the general quantized consensus algorithm proposed by Kashyap et al \cite{Kashyap}. We use the theory of electric networks, random walks, and couplings of Markov chains to derive an $O(N^3\log N)$ upper bound for the expected convergence time on an arbitrary graph of size $N$, improving on the state of art bound of $O(N^5)$ for quantized consensus algorithms. Our result is not dependent on graph topology. Example of complete graphs is given to show how to extend the analysis to graphs of given topology. This is consistent with the analysis in \cite{Drief}.   
\end{abstract}

\let\thefootnote\relax\footnote{This work was presented in part at IEEE INFOCOM 2013.}

\section{Introduction}		
\label{sec:intro}
Over the past decade, the problem of quantized consensus has received  significant attention \cite{Boyd, Benezit, Kashyap, lavaei2012quantized,carli2010gossip, frasca2008average, Cai, cai2012average }. It models averaging in a network with a limited capacity channel \cite{Kashyap}. Distributed algorithms are attractive due to their flexibility, simple deployment and the lack of central control. This problem is of interest in the context of coordination of autonomous agents, estimation, distributed data fusion on sensor networks, peer-to-peer systems, etc. \cite{Drief,Boyd}. It is especially relevant to remote and extreme environments where communication and computation are limited, for example, in a decision-making sensor network \cite{Du}.

This work is motivated by a class of quantized consensus algorithms in \cite{Kashyap}: nodes randomly and asynchronously update local estimate and exchange information. Unlike the distributed algorithm in \cite{Nedic}, where the sum of values in the network is not preserved, Kashyap et al. proposed an algorithm guaranteeing convergence with limited communication, more specifically,  only involving quantization levels \cite{Kashyap}. This is a desired property in a large-scale network where memory is limited, communication between nodes is expensive and no central control is available to the network. Also, this distributed algorithm is designed in a privacy-preserving manner: during the process, the local estimation on the average value is exchanged without revealing the initial observation from nodes. Analysis of convergence time on the complete graph and line graph is given in the original paper in \cite{Kashyap}, and an $O(N^5)$ bound was claimed in \cite{Zhu} by creating a random walk model. 

In this paper, unlike the \emph{natural random walk} model claimed in \cite{Zhu}, we construct a \emph{biased lazy random walk model} for this random communication process to analyze the multi-level quantized consensus problem with the use of Lyapunov functions \cite{Kashyap}. By novelly utilizing the relation between commuting time of a random walk and electric networks \cite{Aldous}, we derive an upper bound on the hitting time of a biased random walk. Several coupled Markov processes are then constructed to help the analysis. Thus we improve the state of art bound in \cite{Zhu} from $O(N^5)$ to $O(N^3\log N)$. In \cite{Drief}, proving through different methods, the authors introduced a function $\delta(G, \alpha)$ depending on the graph structure and voting margin to provide an upper bound on the convergence time of binary consensus algorithm, but did not provide a universal upper bound on an arbitrary graph. Unlike the convergence time bound in \cite{Drief}, which depends on the network topologies and the location of eigenvalues of some contact rate matrices, our result provides a universal upper bound on the convergence time of quantized consensus. Notably, a deterministic protocol was proposed in \cite{hendrickx2011distributed}, which achieves quantized consensus in $O(N^2)$. However, it cannot be extended beyond fixed graphs as the algorithms discussed in this paper, as analyzed in \cite{Zhu}.

The contribution of this paper is as follows:
\begin{itemize}
\item A polynomial upper bound of $O(N^3\log N)$ for the quantized consensus algorithm. It is, to the best knowledge of the authors, the tightest bound in literature for the quantized consensus algorithm proposed in \cite{Benezit,Kashyap}. We use the degree of nodes on the shortest path  on the graph to improve the bound on the hitting time of the biased random walk.       
\item The analysis for arbitrary graphs is extended to a tighter bound for certain network topologies by computing the effective resistance between a pair of nodes on the graph. This is attractive because we can then apply results from algebraic graph theory \cite{Beineke,Godsil} to compute the effective resistance easily on the given graph structure.
\end{itemize}
The remainder of this paper is organized as follows. Section 2 describes the algorithm proposed in \cite{Kashyap}, and formulates the convergence speed problem. In Section 3, we derive our polynomial bound for this class of algorithms. We provide our conclusions in Section 4.


\section{Problem Statement}
\label{sec:problem}

A network is represented by a connected graph $\mathcal{G = (V,E)}$, where $\mathcal{V} =\{1,2,...,N\}$ is the set of nodes and $\mathcal{E}$ is the set of edges. $(i,j)\in \mathcal{E}$ if nodes $i,j$ can communicate with each other. $\mathcal{N}_i$ is the set of neighbors of node $i$. 

Consider a network of $N$ nodes, labeled 1 through $N$. As proposed in \cite{Boyd,Kashyap,Benezit}, each node has a clock which ticks according to a rate 1 exponential distribution. By the superposition property for  the exponential distribution, this set up is equivalent to a single global clock with a rate $N$ exponential distribution ticking at times $\{Z_k\}_{k\ge0}$. The communication and update of  states only occur at  $\{Z_k\}_{k\ge0}$. When the clock of node $i$ ticks, $i$ randomly chooses a neighbor $j$ from the set $\mathcal{N}_i$.  We say edge $(i,j)$ is activated. In the rest of the analysis, for consistency with previous literatures as \cite{Zhu}\cite{Kashyap}, we discretize time instant $t$ according to $\{Z_k\}_{k\ge 0}$,  i.e., in terms of the total number of clock ticks.

In the rest of this section, we will describe the distributed quantized consensus algorithm \cite{Kashyap}. We are interested in the performance of this class of algorithms on arbitrary graphs.

\subsection{Quantized Consensus}
\label{sub:qc}
Without loss of generality, let us assume that all nodes hold integer values and the quantization is 1. Let $Q^{(i)}(t)$ denote the integer value of node $i$ at time $t$, with $Q^{(i)}(0)$ denoting the initial values. Define 
\begin{equation}
\label{qsum}
Q_{\rm{sum}} = \sum_{i  = 1} ^N Q^{(i)}(0).
\end{equation}
Let $Q_{\rm{sum}}$ be written as  $qN + r$, where $ 0 \le r < N$. Then the mean of the initial value in the network $\frac{1}{N} Q_{\rm{sum}} \in [q, q+1)$. Thus either $q$ or $q+1$ is an acceptable integer value for quantized average consensus (if the quantization level is 1). 

\begin{mydef}[Convergence on Quantized Consensus]
A quantized consensus reaches convergence at time $t$, if for any node $i$ on the graph,  $Q^{(i)}(t) \in \{q, q+1\}$.
\end{mydef}

There are a few properties that are desired for the quantized consensus algorithm: 
\begin{itemize}
\item \em{Sum conservation}:  
\begin{equation}
\sum_{i  = 1} ^N Q^{(i)}(t) = \sum_{i  = 1} ^N Q^{(i)}(t+1).
\end{equation}

\item {Variation non-increasing}:  \rm{if two nodes $i$, $j$ exchange information, }
\begin{equation}
|Q^{(i)}(t+1) - Q^{(j)}(t+1)| \le |Q^{(i)}(t) - Q^{(j)}(t)|.
\end{equation}
\end{itemize}

When two nodes $i$ and $j$ exchange information, without loss of generality, suppose that $Q^{(i)}(t)  \le Q^{(j)}(t)$. They follow the simple update rules below:
\begin{enumerate}
\item If $Q^{(j)}(t) - Q^{(i)}(t) \ge 2$, a \emph{non-trivial meeting} occurs:
$$
Q^{(i)}(t+1) = Q^{(i)}(t) + 1, Q^{(j)}(t+1) = Q^{(j)}(t) - 1.
$$
\item If $Q^{(j)}(t) - Q^{(i)}(t) \le 1$, a \emph{trivial meeting} occurs:

$$Q^{(i)}(t+1) = Q^{(j)}(t), Q^{(j)}(t+1) = Q^{(i)}(t).$$

\end{enumerate}

We can view this random process as a finite state Markov chain. Because the variation decreases whenever there is a non-trivial exchange,  convergence will be reached in finite time almost surely.  

\emph{Remark:} In this section, the update rules allow the node values to change by at most 1. This is relevant to load-balancing systems where only one value can be exchanged in the channel at a time due to the communication limit \cite{Kashyap}. Adjustments can be made for this class of quantized consensus algorithms, e.g. when two nodes exchange information, both nodes can update their value to the mean of the two. The analysis on the convergence time remains similar.

\section{Convergence Time Analysis}\label{main}
The main result of this work is the following theorem:
\begin{thm}
For a connected network of N nodes, an upper bound for the expected convergence time of the quantized consensus algorithm is $O(N^3\log(N))$.
\end{thm}

We use the analogy of electric networks and random walks to derive the upper bound. Before deriving the bound on the convergence time, we first provide some definitions and notation that we will use and prove some useful lemmas in   Section \ref{def} and Section \ref{mt}. 

\subsection{Definition and Notation}\label{def}
\begin{mydef}[Hitting Time]
For a graph $\mathcal{G}$ and a specific random walk $X$, and $i, j \in \mathcal{V}$, let $\mathcal{H}{(i,j)}$ denote the expected number of steps a random walk beginning at $i$ must take before reaching $j$, i.e., $\mathcal{H}(i,j) = \mathbf{E}\left[\min\{t: X_t = j \}| X_0 = i \right]$ . Define the ``hitting time'' of $\mathcal{G}$ by $\mathcal{H(G)} = \max_{i,j}\mathcal{H}(i,j)$.
\end{mydef}

\begin{mydef}[Meeting Time]

Consider two random walkers $X,Y$ placed on $\mathcal{G}$, and $i, j \in \mathcal{V}$. At each tick of the clock, they move according to some joint probability distribution. Let $\mathcal{M}{(i,j)}$ denote the expected time for the two walkers starting from $i$ and $j$ respectively to meet at the same node or to cross each other through the same edge (if they move at the same time), i.e.$\mathcal{M}(i,j) = \mathbf{E}\left[\min\{t: X_t = Y_t \textrm{ or } X_t = Y_{t-1},Y_t=X_{t-1} \}| X_0 = i , Y_0 = j\right]$. Define the ``meeting time" of $\mathcal{G}$ by $\mathcal{M(G)} = \max_{i,j}\mathcal{M}(i,j)$.
\end{mydef}

Define a \emph{simple random walk} on $\mathcal{G}$, with transition matrix $P^{S} = (P_{ij})$ as follows:
\begin{itemize}
\item $P^S_{ii}: =0$ for $\forall i \in \mathcal{V}$,
\item $P^S_{ij}: =\frac{1}{|\mathcal{N}_i|}$ for $(i, j) \in \mathcal{E}$.  
\end{itemize}
$\mathcal{N}_i$ is the set of neighbors of node $i$ and $|\mathcal{N}_i|$ is the degree of node $i$.

Define a \emph{natural random walk}  with transition matrix $P^{N} = (P_{ij})$ as follows:
\begin{itemize}
\item $P^N_{ii}=1-\frac{1}{N}$ for $\forall i \in \mathcal{V}$,
\item $P^N_{ij}=\frac{1}{N|\mathcal{N}_i|}$ for $(i, j) \in \mathcal{E}$.
\end{itemize}

Define a \emph{biased random walk}  with transition matrix $P^B = (P_{ij})$ as follows:
\begin{itemize}
\item $P^B_{ii}: =1-\frac{1}{N} - \sum_{k\in \mathcal{N}_i}\frac{1}{N|\mathcal{N}_k|}$ for $\forall i \in \mathcal{V}$,
\item $P^B_{ij}: =\frac{1}{N}\left(\frac{1}{|\mathcal{N}_i|} + \frac{1}{|\mathcal{N}_j|}\right)$ for $(i, j) \in \mathcal{E}$.
\end{itemize}

\subsection{Hitting Time and Meeting Time on Weighted Graph}\label{mt}

In this class of algorithms, we label the initial observations(states or values) by the nodes as $\alpha_1, \alpha_2, ..., \alpha_N$. A random walk is a Markov process with random variables $A_1, A_2,..., A_t,...$ such that the next state only depends on the current state. In the system setting, when the node $i$'s clock ticks, $i$ randomly choose one of its neighbor node $j$ from the set $\mathcal{N}_i$ to exchange information. We notice that before any two observations $\alpha_m, \alpha_n$ meet each other,  they take random walks on the graph $\mathcal{G}$. Their marginal transition matrices are both $P^B$. It may be tempting to think that they are taking the \emph{natural random walks} as stated in \cite{Zhu}. Upon closer inspection, we find that there are two sources stimulating the random walk from $i$ to $j$, for all $(i,j)\in \mathcal{E}$: one is active, initiated by node $i$'s clock, which leads to $P^1_{ij} = P^N_{ij}$; the other one is passive, initiated by $i$'s neighbor $j$, which leads to $P^2_{ij} = P^N_{ji}$. Thus $P_{ij} = P^1_{ij} + P^2_{ij}= P^B_{ij}$; i.e., the transition matrix is actually $P^B$ instead of $P^N$. Because of the system settings, two random walks $\alpha_m, \alpha_n$ can only move at the same time if they are adjacent. Denote this joint random process as $\mathcal{X}$. Suppose $\alpha_m$ is at node $x$, and $\alpha_n$ is at node $y$.

For $x \notin \mathcal{N}_y$,  and  $i \in \mathcal{N}_x$, we have
\begin{eqnarray}\nonumber 
&& P_{\mathcal{X}\textrm{joint}}(\alpha_m\textrm{   moves from } x \textrm{ to } i,\textrm{ } \alpha_n\textrm{ does not move}) \\ \nonumber &= &
P^B_{xi} - P_{\mathcal{X}\textrm{joint}}(\alpha_m\textrm{ moves from } x \textrm{ to } i\textrm{, } \alpha_n\textrm{ moves}) \\ 
& = &P^B_{xi}.
\label{bayes}
\end{eqnarray}
Similar for $P_{\mathcal{X}\textrm{joint}}(\alpha_n\textrm{   moves from } y \textrm{ to } j ,\textrm{ } \alpha_m\textrm{ does not move})$, where $j \in \mathcal{N}_y$. Also, 

\begin{eqnarray} \nonumber
&& P_{\mathcal{X}\textrm{joint}}(\alpha_m\textrm{ does not move, } \alpha_n\textrm{ does not move}) \\
&=& 1 - \sum_{i\in \mathcal{N}_x}P^B_{xi} - \sum_{j\in \mathcal{N}_y}P^B_{yj}.
\end{eqnarray}

For $x \in \mathcal{N}_y$ and $i \neq y$ we have,  
\begin{eqnarray}
\nonumber 
&& P_{\mathcal{X}\textrm{joint}}(\alpha_m\textrm{   moves from } x \textrm{ to } i,\textrm{ } \alpha_n\textrm{ does not move}) \\ \nonumber &= &
P^B_{xi} - P_{\mathcal{X}\textrm{joint}}(\alpha_m\textrm{ moves from } x \textrm{ to } i\textrm{, } \alpha_n\textrm{ moves})  \\
 &=& P^B_{xi}.
\end{eqnarray}
\begin{eqnarray} 
 &&P_{\mathcal{X}\textrm{joint}}(\alpha_m\textrm{ moves to }y, \alpha_n \textrm{ moves to }x) = P^B_{xy}.
 \label{meet}
\end{eqnarray}

\begin{eqnarray}\nonumber 
&&P_{\mathcal{X}\textrm{joint}}(\alpha_m\textrm{ does not move, } \alpha_n\textrm{ does not move}) \\
&=& 1 - \sum_{i\in \mathcal{N}_x}P^B_{xi} - \sum_{j\in \mathcal{N}_y}P^B_{yj} + P^B_{xy}.
\label{stay}
\end{eqnarray}

\begin{lem}
The biased random walk $\mathcal{X}$ is a reversible Markov process.
\end{lem}
\begin{proof}
Let $\pi$ be the stationary distribution of the biased random walk $\mathcal{X}$. It is easy to verify that
\begin{equation}
\pi_i = \frac{1}{N}
\end{equation}
 for all $i\in \mathcal{V}$.
Thus by the symmetry of $P^B$, $$\pi_iP^B_{ij} = \pi_jP^B_{ji}.$$
\end{proof}

\begin{lem} \label{lem:hitting}
In an arbitrary connected graph $\mathcal{G}$ with $N$ nodes, the hitting time of the biased random walk $\mathcal{X}$ satisfies $$\mathcal{H}_{P^B}(\mathcal{G}) < 3{N^3}.$$
\end{lem}
\begin{proof}
The biased random walk $\mathcal{X}$ defined above is a random walk on a weighted graph with weight
\begin{equation}\label{weight}
w_{ij}: = \frac{1}{N}\left(\frac{1}{|\mathcal{N}_i|} + \frac{1}{|\mathcal{N}_j|}\right) \textrm{ for } (i, j) \in \mathcal{E}.
\end{equation}
\begin{equation}
w_{ii} : = 1 - \sum_{j \in \mathcal{N}_i}w_{ij}.
\end{equation} 
\begin{equation}
w_i = \sum_{j\in \mathcal{V}}w_{ij} = 1, \; w = \sum_i{w_i} = N. 
\end{equation}

\begin{figure}[!t]
\centering
\centerline{\includegraphics[width=.4\textwidth]{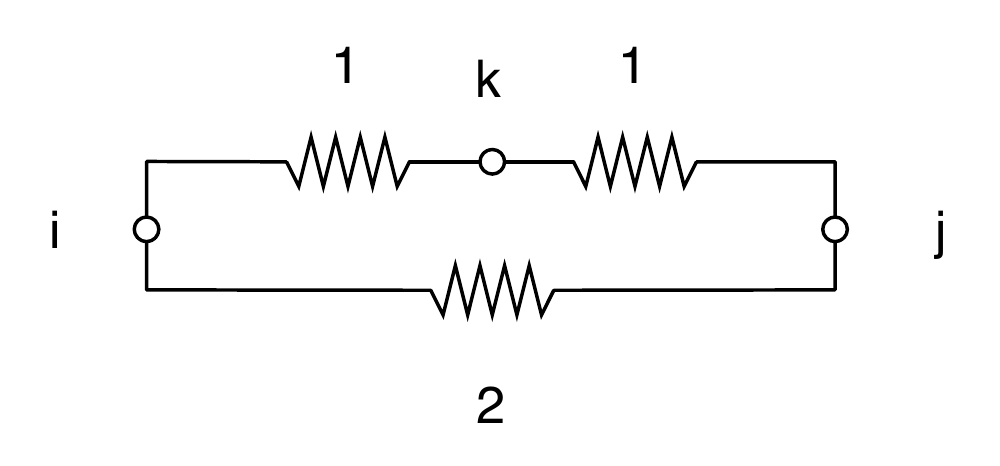}}
\caption{{\em A simple example of effective resistance.}}
\label{resistance}
\end{figure}

It is well-known that there is an analogy between a weighted graph and an electric network \cite{Aldous}. Let $r_{ij}$ denote the resistance between to adjacent nodes, i.e. an edge $(i, j) \in \mathcal{E}$, and let $r'_{xy} $ denote the effective resistance between any two nodes $x, y$. For example, in Fig. \ref{resistance}, $r_{ij} = 2$, $r_{ik} = 1$, $r'_{ij} = 1$ and $r'_{ik} = 0.75$.   In an electric circuit, we always have $r'_{ij} \le r_{ij}$ because of Ohm's Law. For a random walk on a weighted graph, a wire linking $i$ and $j$ has conductance $w_{ij}$, i.e., resistance $r_{ij} = 1/w_{ij}$. And the commuting time between $x$ and $y$, $\mathcal{H}_{P^B}(x, y) + \mathcal{H}_{P^B}(y, x)$, has the following relationship with the effective resistance $r'_{xy}$:
\begin{equation}\label{commute}
\mathcal{H}_{P^B}(x, y) + \mathcal{H}_{P^B}(y, x) = wr'_{xy},
\end{equation}
where $r'_{xy}$ is the effective resistance in the electric network between node $x$ and node $y$ (Chapter 3 Corollary 11 in \cite{Aldous}). Since the degree of any node is at most $N-1$, for $(i, j) \in \mathcal{E}$, 
\begin{equation}
w_{ij} = \frac{1}{N}\left(\frac{1}{|\mathcal{N}_i|} + \frac{1}{|\mathcal{N}_j|}\right) 
  >  \frac{1}{N}\frac{1}{\min(|\mathcal{N}_i|, |\mathcal{N}_j|)} 
\end{equation}

\begin{equation}
\label{h1}
r_{ij} < N\times \min(|\mathcal{N}_i|, |\mathcal{N}_j|)
\end{equation}

Consequently, $r'_{ij} \le r_{ij} <N\times \min(|\mathcal{N}_i|, |\mathcal{N}_j|)$. 

For all $x, y \in \mathcal{V}$, let $Q = ( q_1 = x, q_2, q_3, ..., q_{l-1}, q_l = y)$ be the shortest path on the graph connecting $x$ and $y$ . Now we claim that 
$\sum^l _{k = 1}|\mathcal{N}_{q_k}| < 3N$ (from the proof of Theorem 2 in \cite{ikeda2009hitting}).

Since any node not lying on the shortest path can only be adjacent to at most three vertices on $Q$, we have
\begin{equation}
\label{h2}
\sum^l _{k = 1}|\mathcal{N}_{q_k}| \le 2l + 3(N - l) < 3N.
\end{equation}
The first term $2l$ in Equation (\ref{h2}) is due to the fact that nodes on the shortest path connect to one another, resulting in total degree about $2l$ ($2(l-1)$ to be precise). The second term $3(N-l)$ is due to the fact that any node, say $u$, not lying on the shortest path can only be adjacent to at most three vertices on the shortest path, say $q_{i1}, q_{i2}, q_{i3}$ on $Q$. Suppose $u$ is also adjacent to $q_{i4}$, which is also on $Q$. Without loss of generality, let $i_1 < i_2 < i_3 < i_4$. Then the path $q_1, ..., q_{i1},u,q_{i4},...,q_l$ is shorter than the path $Q=(q_1,...,q_{i1},...,q_{i2},...,q_{i3},...,q_{i4},...,q_l)$, contradict with the fact $u$ is not on shortest path. Thus the $(N-l)$ points not on the shortest path contributing at most $3(N-l)$ degrees for the total degrees of nodes on the shortest path. 

The effective resistance between any two nodes $x$ and $y$ is less than or equal to the sum of the effective resistance $r'_{q_kq_{k+1}}$ on the shortest path.  This is due to the triangle inequality for effective resistance on undirected graphs (Theorem B in \cite{klein1993resistance}). By (\ref{h1}) and (\ref{h2}), we have 
\begin{eqnarray}\nonumber
r'_{xy} &\le& \sum_{k = 1}^{l-1} r'_{q_kq_{k+1}}
\le \sum_{k = 1}^{l-1} N\times \min(|\mathcal{N}_{q_k}|, |\mathcal{N}_{q_k+1}|) \\
&\le& N \times \sum^l _{k = 1}|\mathcal{N}_{q_k}|  
< 3N^2
\end{eqnarray}

By (\ref{commute}), we have 
\begin{eqnarray}\nonumber 
\mathcal{H}_{P^B}(x,y) &<& \mathcal{H}_{P^B}(x, y) + \mathcal{H}_{P^B}(y, x) \\ \nonumber
&=& wr'_{xy}  \\ \nonumber
&<& N \times 3N^2  \\
&=& 3N^3.
\end{eqnarray}
This completes the proof.
\end{proof}
Note that this is an upper bound for arbitrary connected graphs. A tighter bound can be derived for certain network topologies. 

\begin{lem}\label{circle}
$\mathcal{H}_{P^B}(x, y) + \mathcal{H}_{P^B}(y, z) + \mathcal{H}_{P^B}(z, x) = \mathcal{H}_{P^B}(x, z) + \mathcal{H}_{P^B}(z, y) + \mathcal{H}_{P^B}(y, x). $
\end{lem}
\begin{proof}
This is direct result from Lemma 2 in Chap 3 of Aldous-Fill's book \cite{Aldous} since $\mathcal{X}$ is reversible. 
\end{proof}

\begin{mydef}[Hidden Vertex]
A vertex $t$ in a graph is said to be hidden if for every other point in the graph,  $\mathcal{H}(t, v) \le\mathcal{H}(v, t)$. A hidden vertex is shown to exist for all reversible Markov chains in Lemma 3 in \cite{Coppersmith}.
\end{mydef} 

\begin{lem}\label{lem:meeting}
The meeting time of any two random walders on the network $\mathcal{G}$ following the random processes $\mathcal{X}$ in Section \ref{mt} is less than $4\mathcal{H}_{P^B}{(\mathcal{G})}$.
\end{lem}
\begin{proof}
In order to prove the lemma, we construct a coupled Markov chain, $\mathcal{X}'$ to assist the analysis. $\mathcal{X}'$ has the same joint distribution as $\mathcal{X}$  except (\ref{meet}) and (\ref{stay}).
\begin{eqnarray} 
 P_{\mathcal{X}'\textrm{joint}}(\alpha_m\textrm{, }\alpha_n \textrm{ meet at }x \textrm{ or }y) = 2P^B_{xy}
 \label{meet2}
\end{eqnarray}
\begin{eqnarray} \nonumber
&& P_{\mathcal{X}'\textrm{joint}}(\alpha_m\textrm{ does not move, } \alpha_n\textrm{ does not move}) \\
&=& 1 - \sum_{i\in \mathcal{N}_x}P^B_{xi} - \sum_{j\in \mathcal{N}_y}P^B_{yj}.
\label{stay2}
\end{eqnarray}

The proof is based on the following sequence of claims:
\begin{enumerate}
\item The meeting time of two random walkers following $\mathcal{X}'$ is less than $2\mathcal{H}_{P^B}{(\mathcal{G})}$. 
\item  The meeting time of random process $\mathcal{X}$ and  $\mathcal{X}'$ satisfies $\mathcal{M}_{\mathcal{X}}(\mathcal{G}) \leq 2\mathcal{M}_{\mathcal{X}'}(\mathcal{G})$.
\item There holds that  $\mathcal{M}_{\mathcal{X}}(\mathcal{G}) < 4\mathcal{H}_{P^B}{(\mathcal{G})}$.
\end{enumerate}

The formal proof is as follows:

For convenience, we adopt the following notation: Let $\mathcal{H}(\bar{x},y)$ denote the weighted average of $\mathcal{H}(u,y)$ over all neighbors $u$ of $x$; Let $\mathcal{M}(\bar{x},y)$ denote the weighted average of $\mathcal{M}(u,y)$ over all neighbors $u$ of $x$; Let $\phi(\bar{x},y)$ denote the weighted average of $\phi(u,y)$ over all neighbors $u$ of $x$. Weightings are according to the edge weights.

Similar as in \cite{Coppersmith}, define a \emph{potential function}
\begin{equation}
\label{eqn:potential}
\phi(x,y) := \mathcal{H}_{P^B}(x,y) + \mathcal{H}_{P^B}(y, t) - \mathcal{H}_{P^B}(t,y),
\end{equation}
 where $t$ is a hidden vertex on the graph. By Corollary \ref{circle}, $\phi(x,y)$ is symmetric, i.e. $\phi(x,y) = \phi(y,x)$. By the definition of meeting time, $\mathcal{M}$ is also symmetric, i.e. $\mathcal{M}(x,y) = \mathcal{M}(y,x)$. Next we use $\phi$ to bound the meeting time. 
 
 By the definition of hitting time, for $x \neq y$ we have
\begin{eqnarray} \nonumber
\mathcal{H}_{P^B}(x, y) 
&=& 1 + P^B_{xx}\mathcal{H}_{P^B}(x,y) + \sum_{i\in \mathcal{N}_x}P^B_{xi}\mathcal{H}_{P^B}(i,y)\\\nonumber
&=&1 + w_{xx}\mathcal{H}_{P^B}(x,y) + \sum_{i\in \mathcal{N}_x}w_{xi}\mathcal{H}_{P^B}(i,y), \\
\end{eqnarray}
i.e.,
\begin{eqnarray}\nonumber \label{eqn:p2}
\mathcal{H}_{P^B}(x, y) &=& \frac{1}{\sum_{i\in \mathcal{N}_x}w_{xi}} + \frac{\sum_{i\in \mathcal{N}_x}w_{xi}\mathcal{H}_{P^B}(i,y)}{\sum_{i\in \mathcal{N}_x}w_{xi}} \\
&=& \frac{1}{\sum_{i\in \mathcal{N}_x}w_{xi}} + \mathcal{H}(\bar{x},y).
\end{eqnarray}
So for $x \neq y$, by Equation (\ref{eqn:potential}) and (\ref{eqn:p2}),
\begin{equation} \label{phi}
\phi(x,y) = \frac{1}{\sum_{i\in \mathcal{N}_x}w_{xi}} + \phi(\bar{x},y).
\end{equation}
Similarly, by the definition of meeting time, we have,
\begin{eqnarray}\nonumber 
\mathcal{M}_{\mathcal{X}'}(x, y)  &= & 1 + \left( 1 - \sum_{i \in \mathcal{N}_x}P^B_{xi} - \sum_{j \in \mathcal{N}_y}P^B_{yj}\right)\mathcal{M}_{\mathcal{X}'}(x,y) \\ \nonumber 
& + & \sum_{i \in \mathcal{N}_x}P^B_{xi}\mathcal{M}_{\mathcal{X}'}(i,y) \\
& + & \sum_{j \in \mathcal{N}_y}P^B_{yj}\mathcal{M}_{\mathcal{X}'}(x,j).
\label{universal}
\end{eqnarray}
Note that   (\ref{universal}) also holds for $x\in \mathcal{N}_y$. We now have

\begin{eqnarray}\nonumber
 &&\left( \sum_{i \in \mathcal{N}_x}P^B_{xi} + \sum_{j \in \mathcal{N}_y}P^B_{yj}\right)\mathcal{M}_{\mathcal{X}'}(x,y)  \\
 &=&   1 +  \sum_{i \in \mathcal{N}_x}P^B_{xi}\mathcal{M}_{\mathcal{X}'}(i,y) 
 +  \sum_{j \in \mathcal{N}_y}P^B_{yj}\mathcal{M}_{\mathcal{X}'}(x,j).
 \label{ineq}
\end{eqnarray}

  (\ref{ineq}) shows that at least one of the two inequalities below holds:
\begin{equation}
\mathcal{M}_{\mathcal{X}'}(x, y) > \frac{\sum_{i \in \mathcal{N}_x}P^B_{xi}\mathcal{M}_{\mathcal{X}'}(i,y)}{\sum_{i \in \mathcal{N}_x}P^B_{xi}} = \mathcal{M}_{\mathcal{X}'}(\bar{x},y)
\end{equation}
\begin{equation}
\mathcal{M}_{\mathcal{X}'}(x, y) > \frac{\sum_{j \in \mathcal{N}_y}P^B_{yj}\mathcal{M}_{\mathcal{X}'}(x,j)}{\sum_{j \in \mathcal{N}_y}P^B_{yj}} = \mathcal{M}_{\mathcal{X}'}(x,\bar{y})
\label{yineq}
\end{equation}
Without loss of generality, suppose that (\ref{yineq}) holds (otherwise, we can prove the other way around). From (\ref{ineq}), we have

\begin{eqnarray}\nonumber 
&& \sum_{i \in \mathcal{N}_x}P^B_{xi} \mathcal{M}_{\mathcal{X}'}(x,y)  
  =    1 +  \sum_{i \in \mathcal{N}_x}P^B_{xi}\mathcal{M}_{\mathcal{X}'}(i,y) \\
 & + &  \sum_{j \in \mathcal{N}_y}P^B_{yj}\mathcal{M}_{\mathcal{X}'}(x,j) - \sum_{j \in \mathcal{N}_y}P^B_{yj}\mathcal{M}_{\mathcal{X}'}(x,y).
\end{eqnarray}

i.e.,
\begin{eqnarray}\nonumber  
\mathcal{M}_{\mathcal{X}'}(x,y)  
   & = &   \frac{1}{\displaystyle\sum_{i \in \mathcal{N}_x}P^B_{xi}} +  \mathcal{M}_{\mathcal{X}'}(\bar{x},y) \\\nonumber
& + &  \frac{\displaystyle\sum_{j \in \mathcal{N}_y}P^B_{yj}\left(\mathcal{M}_{\mathcal{X}'}(x,\bar{y}) - \mathcal{M}_{\mathcal{X}'}(x,y)\right)}{\displaystyle\sum_{i \in \mathcal{N}_x}P^B_{xi}} \\
& <  & \frac{1}{\displaystyle\sum_{i \in \mathcal{N}_x}w_{xi}} + \mathcal{M}_{\mathcal{X}'}(\bar{x},y).
\label{lastineq}
\end{eqnarray}
Now we claim that $\mathcal{M}_{\mathcal{X}'}(x,y) \le \phi(x,y)$. Suppose it is not the case. Let $\beta = \max_{x,y}\{\mathcal{M}_{\mathcal{X}'}(x,y) - \phi(x,y)\}$. Among all the pairs $x, y$ realizing $\beta$, choose any pair.
It is clear that $x \neq y$, since $\mathcal{M}_{\mathcal{X}'}(x,x) = 0 \le \phi(x,x)$. By (\ref{phi}) and (\ref{lastineq}), 
\begin{eqnarray}\nonumber
\mathcal{M}_{\mathcal{X}'}(x,y) &=& \phi(x,y) + \beta 
= \frac{1}{\sum_{i\in \mathcal{N}_x}w_{xi}} + \phi(\bar{x}, y) + \beta\\  \nonumber
&\ge& \frac{1}{\sum_{i\in \mathcal{N}_x}w_{xi}} + \mathcal{M}_{\mathcal{X}'}(\bar{x},y) \\
 &>& \mathcal{M}_{\mathcal{X}'}(x,y).
\end{eqnarray}
This is a contradiction. From the definition of the \emph{potential} function in (\ref{eqn:potential}), we have $\phi(x,y) < 2\mathcal{H}_{P^B}{(\mathcal{G})}$. 
Thus $\mathcal{M}_{\mathcal{X}'}(\mathcal{G}) \le  \phi(x,y) < 2\mathcal{H}_{P^B}{(\mathcal{G})}$.

Now we are ready to complete the proof of Lemma \ref{lem:meeting}. We couple the Markov chains $\mathcal{X}$ and $\mathcal{X'}$ so that they are equal until the two random walkers become neighbors. Note that half of the time in expectation almost surely when the walkers in $\mathcal{X}'$ meet, they do not meet in  $\mathcal{X}$, but stay in the same position. We claim that $\mathcal{M}_{\mathcal{X}}(\mathcal{G}) \leq 2\mathcal{M}_{\mathcal{X}'}(\mathcal{G})$. 

In the random process $\mathcal{X}'$, when two random walkers $m$, $n$ meet, instead of finishing the process, we let them exchange positions and continue the random walks according to $P_{\mathcal{X}'\textrm{joint}}$. The expected length of each exchange is less than or equal to $\mathcal{M}_{\mathcal{X}'(\mathcal{G})}$. At each cross, the random process $\mathcal{X}$ finishes with a probability of 1/2, independently. Thus for any $x, y \in \mathcal{V}$ we have 
\begin{equation}
\mathcal{M}_{\mathcal{X}}(x,y) \leq \sum_{i = 1}^{\infty}\left(\frac{1}{2}\right)^ii\mathcal{M}_{\mathcal{X}'}(\mathcal{G}) = 2\mathcal{M}_{\mathcal{X}'}(\mathcal{G}).
\end{equation}

This completes the proof.
\end{proof}

Let $CT_{\mathcal{G}}(v)$ denote the expected time for a random walker starting from node $v$ to meet all other random walkers who are also taking random walks on the same graph but starting from different nodes. Define 
$CT(\mathcal{G}) = \max_{v \in \mathcal{V}} CT_{\mathcal{G}}(v).$

\begin{lem}\label{max}
Let $\mathcal{M}_{\mathcal{X}}(\mathcal{G})$ be the meeting time of the biased random walk $\mathcal{X}$ defined in Section \ref{def}. Then
\begin{equation}
\label{ }
CT(\mathcal{G}) = O(\mathcal{M}_{\mathcal{X}}(\mathcal{G})\log N).
\end{equation}\end{lem}
\begin{proof}
Since there are no more than $N$ consecutive meetings with random walkers that it never meets, we can easily get a union bound for $CT(\mathcal{G})$, which is $N\mathcal{M}_{\mathcal{X}}(\mathcal{G})$.

In order to obtain a tighter bound for $CT(\mathcal{G})$, we divide the random walk into $\ln N$ periods of length $k\mathcal{M}_{\mathcal{X}}(\mathcal{G})$ each, where $k$ is a constant. Let $a$ be the ``special" random walker trying to meet all other random walkers. For any period $i$ and any other random walker $v$, by the Markov inequality, we have
\begin{equation}
\Pr(\textrm{$a$ does not meet $v$ during period $i$})  
\le \frac{\mathcal{M}_{\mathcal{X}}(\mathcal{G})}{k\mathcal{M}_{\mathcal{X}}(\mathcal{G})} 
= \frac{1}{k}
\end{equation}
so
\begin{equation}
\Pr(\textrm{$a$ does not meet $v$ during any period})  
\le\left( \frac{1}{k}\right)^{\ln N} = N^{-\ln k}
\end{equation}
If we take the union bound, 
\begin{equation}
\Pr(\textrm{$a$ doesn't meet some walker during any period}) 
 \le  N\cdot N^{-\ln k}.
\end{equation}

Conditioning on whether or not the walker $a$ has met all other walkers after all $k\mathcal{M}_{\mathcal{X}}(\mathcal{G})\ln N$ steps, and using the previous $N\mathcal{M}_{\mathcal{X}}(\mathcal{G})$ upper bound, we have

\begin{eqnarray}\nonumber
CT(\mathcal{G}) &\le& k\mathcal{M}_{\mathcal{X}}(\mathcal{G})\ln N + N\cdot N^{-\ln k}\cdot N\mathcal{M}_{\mathcal{X}}(\mathcal{G})\\ 
&=& k\mathcal{M}_{\mathcal{X}}(\mathcal{G})\ln N +  N^{2-\ln k}\mathcal{M}_{\mathcal{X}}(\mathcal{G})
\end{eqnarray}

When $k$ is sufficiently large, say $k \ge e^6$, the second term is small, so
\begin{equation}
\label{ }
CT(\mathcal{G}) < (k + 1)\mathcal{M}_{\mathcal{X}}(\mathcal{G})\ln N.
\end{equation}
This completes the proof. 
\end{proof}

\subsection{An Upper Bound on Quantized Consensus} \label{boundqc}
Recall that a non-trivial exchange in quantized consensus happens when the difference in values at the nodes is greater than 1. 

Let $\mathbf{Q}(t)$ denote a vector of values all nodes holding at time $t$. Set $\bar{Q} = Q_{\rm{sum}}/N$, where $Q_{\rm{sum}}$ is defined in   (\ref{qsum}).

We construct a Lyapunov function $L_{\bar{Q}}$ \cite{Zhu,Kashyap,Nedic} as:

\begin{equation}
L_{\bar{Q}}(\mathbf{Q}(t)) = \sum_{i = 1}^N \left(Q_i(t) - \bar{Q}\right)^2.
\end{equation} 

Let $m = \min_i Q_i(0)$ and $M = \max_i Q_i(0)$. It is easy to see that $L_{\bar{Q}}(\mathbf{Q}(0)) \le \frac{(M-m)^2N}{4}$. Equality holds when half of the values are $M$ and others are $m$.

\begin{lem}
\label{cor:lfunc}
In a non-trivial meeting, $$L_{\bar{Q}}(\mathbf{Q}(t)) \ge L_{\bar{Q}}(\mathbf{Q}(t+1)) + 2.$$
\end{lem}

\begin{proof}
A non-trivial meeting follows the first update rule of quantized consensus algorithm in Section \ref{sub:qc}. 

Suppose $Q_i(t) = x_1$ and $Q_j(t) = x_2$ have a non-trivial meeting at time $t$, and the rest of the values stay unchanged. Without loss of generality, let $x_1 \le x_2 - 2$.  We have

\begin{eqnarray} \nonumber
&&L_{\bar{Q}}(\mathbf{Q}(t)) - L_{\bar{Q}}(\mathbf{Q}(t+1)) \\ \nonumber
&=& x_1^2 +   x_2^2 - (x_1 + 1)^2 - (x_2 - 1)^2 \\ 
&=& 2(x_2 - x_1) - 2 \ge 2.
\end{eqnarray}
\end{proof}

\begin{proof}[Proof of Theorem 1]
Corollary \ref{cor:lfunc} shows that the Lyapunov function is decreasing. The convergence of quantized consensus must be reached after at most $\gamma = \frac{(M-m)^2N}{8}$ \emph{non-trivial meetings}. When every random walker has met each of the other random walkers  $\gamma$ times, all the \emph{non-trivial meetings} must have finished. The rest of the proof follows from Corollary \ref{max}, except we divide the random walks into $\ln(\gamma N/2)$ periods of length $k\mathcal{M}_{\mathcal{X}}(\mathcal{G})$ instead of $\ln N$. By Corollary \ref{lem:hitting}, \ref{lem:meeting} and  \ref{max}, this process finishes in $O(N^3\log N)$ time. 
\end{proof}

\begin{cor}
\label{cor:tighter}
Given a network  $\mathcal{G}$ of N nodes, an upper bound for the expected convergence time of the quantized consensus algorithm is $O(\mathcal{M}_{\mathcal{X}}(\mathcal{G})\log N)$. 
\end{cor}
\begin{proof}
This is direct result from the proof of Corollary \ref{max} and Theorem 1.
\end{proof}

For example, for a fully connected network $\mathcal{C}$ (i.e.complete graph) with $N$ nodes, at time $t$, the probability of any two random walker $i$, $j$ to meet is $P^B(\mathcal{C}) = \frac{2}{N(N-1)}$. It is easy to get that $\mathcal{M}_{P^B}(\mathcal{C}) = \frac{N(N-1)}{2} = O(N^2)$. Thus by Corollary \ref{cor:tighter}, the upper bound for the expected convergence time of the quantized consensus algorithms is $O(N^2\log N)$. Note that this result agrees with the analysis of convergence time of complete graph in Section IV.A in \cite{Drief}, where the authors derived an upper bound of $O(N\log N)$, regarding to local clock (See Section 2 for definition of local clock and global clock). Since every second, there are number of $N$ clock ticks on average, this is hence equivalent to a  $O(N^2\log N)$ bound regarding to the number of clock ticks in our case. More examples and simulations can be found in \cite{full}.

\section{conclusions}
\label{sec: discussion}
In this paper, we use the theory of electric networks, random walks, and couplings of Markov chains to derive a polynomial bound on convergence time with respect to the size of the network, for a class of distributed quantized consensus algorithms \cite{Benezit,Kashyap}. We improve the state of art bound of $O(N^5)$ for quantized consensus algorithms to $O(N^3\log N)$. Our analysis can be extended to a tighter bound for certain network topologies using the effective resistance analogy.  Our results provide insights to the performance of the quantized consensus algorithms.



\bibliographystyle{IEEEbib}
\bibliography{myrefs}

\end{document}